\documentclass[final,5p,times,authoryear]{elsarticle}
\usepackage{amsmath,amssymb,amsthm,mathtools}
\usepackage{microtype}
\usepackage{hyperref}
\journal{Automatica}
\biboptions{round,semicolon}

\hypersetup{
  colorlinks=true,
  linkcolor=blue,
  citecolor=blue,
  urlcolor=blue,
  pdftitle={Elementwise Positivity of the Solution to Lyapunov Equation for Hurwitz Companion
Matrices}
}

\newtheorem{theorem}{Theorem}[section]
\newtheorem{lemma}[theorem]{Lemma}
\newtheorem{proposition}[theorem]{Proposition}
\newtheorem{corollary}[theorem]{Corollary}
\theoremstyle{definition}
\newtheorem{example}[theorem]{Example}
\theoremstyle{remark}
\newtheorem{remark}[theorem]{Remark}

\begin{document}

\begin{frontmatter}

\title{Elementwise Positivity of the Solution to Lyapunov Equation for Hurwitz Companion Matrices}

\author[aff1]{Miaomiao Wang}
\author[aff2]{Patrizio Colaneri}
\author[aff1]{Jie Chen}
\ead{jichen@cityu.edu.hk}
% \ead{miaomiao.wang@cityu.edu.hk; patrizio.colaneri@polimi.it; jichen@cityu.edu.hk}
%\thanks{ }
\address[aff1]{Department of Electrical Engineering, City University of Hong Kong, Hong Kong SAR, China} 
% (email: miaomiao.wang@cityu.edu.hk; jichen@cityu.edu.hk}
\address[aff2]{Department of Electronics, Information and Bioengineering, Polytechnic University of Milan, Italy}

\begin{abstract}
We prove, with the aid of AI, that for every real symmetric forcing matrix $Q\succeq0$, the
unique solution of a continuous-time Lyapunov equation is entrywise
nonnegative whenever the state matrix is a real Hurwitz companion
matrix.  This proves an earlier conjecture. The proof makes no assumption on the spectrum of the state matrix, and is made possible by means of using Horner polynomial matrices. 
%
%resolves the companion
%Lyapunov entrywise-nonnegativity conjecture without a real-spectrum
%assumption and without relying on a spectral decomposition.  
% The main
% ingredient is a mixed positivity theorem for the coefficients of the
% adjugate polynomial of an accretive matrix. 
%The proof establishes a mixed positivity property for the coefficients
%of the adjugate polynomial of a real accretive matrix. By dimension %reduction, decomposes a compressed kernel
%into a scalar determinant product and a lower-dimensional quadratic
%term.  
% We give two complementary
% proofs.  The first, by dimension reduction, decomposes a compressed kernel
% into a scalar determinant product and a lower-dimensional quadratic
% term.  The second constructs a scalar rational zero-exclusion
% certificate that makes the kernel right-half-plane stable and then
% invokes a coefficientwise positivity lemma. 
%We further show that
%positive-definite forcing makes
%every entry of the continuous-time solution strictly positive.  
We also show that discrete-time counterpart of the conjecture is false: Counterexamples can be readily constructed to show that Schur companion
matrices, including one generated by a polynomial with all positive
coefficients, yield solutions with negative off-diagonal entries
even with identity forcing.  
%We show that discrete time retains Loewner
%positivity and, after a Cayley transformation followed by a companion
%realization, entrywise positivity in a transformed coordinate system.
\end{abstract}


\end{frontmatter}

\section{Introduction}
\label{sec:introduction}

For a Hurwitz matrix $A$ and a symmetric matrix $Q\succeq0$, the
continuous-time Lyapunov equation
\begin{equation}
  XA+A^{\mathsf{T}}X=-Q
  \label{eq:intro-continuous}
\end{equation}
has a unique solution $X\succeq0$.  This familiar statement concerns
the Loewner order.  It does not, by itself, concerns the signs of the
individual entries of $X$. Nor is it true in general that any of such holds for a general Hurwitz matrix. This paper is concerned with Hurwitz matrices given in the companion form.

Lyapunov, Sylvester, and Stein equations with companion state matrices
have been studied through factored forms, Bezoutian and Hankel
structure, and canonical-form solution algorithms; see
\citet{LancasterLererTismenetsky1984,HeinigJungnickel1986,SreeramAgathoklis1991}
for continuous-time matrix equations and
\citet{KanellakisTawfikAgathoklis2021} for a discrete-time treatment.
Those works provide structural or explicit solution methods, whereas the
present question concerns the signs of all entries of the solution
operator under arbitrary positive-semidefinite forcing.

There was for some time a conjecture, largely circulated through word of mouth among the Italian controls community but seemingly unknown to the outside world, concerning the nonnegativity of the elements in $X$. The conjecture was announced formally in a recent article by  
\citet{Ferrante2026}, which was attributed to Luigi Fortuna, a fomer control theorist at the University of Catania, Italy. 
%
%For the controller companion realization of a real Hurwitz
%polynomial, \citet{Ferrante2026} formulated a conjecture, attributed
%there to Luigi Fortuna, 
Ferrante proved partially that with a Hurwitz companion matrix $A$, the solution of
\eqref{eq:intro-continuous} is entrywise nonnegative for every
$Q\succeq0$. More specifically, he proved the conjecture under a real-spectrum
assumption (all the eigenvalues of $A$ are real), using positivity results for a class of Cauchy-like matrices
\citep{FerranteCauchy2026}. Ferrante's proof is exclusively built on the fact that a companion matrix can be diagonalized by a Vandermonde matrix if the matrix has distinct eigenvalues, and otherwise transformed into a Jordan form by a confluent Vandermonde matrix. 

Our first purpose in this paper is to provide a complete proof, that is, we prove the conjecture without the real-spectrum assumption. By the joint effort of the authors and AI tools, we came up with an elementary proof that makes no distinction of companion matrices with or without complex eigenvalues, and unlike in \citet{Ferrante2026}, has no reference to Vandermonde matrices. Instead, the proof utilizes only the companion matrix structure and the fact that the Lyapunov equation admits a unique semidefinite solution, the bare minimum information known from the problem formulation. Specifically, it proceeds with the construction of Horner matrix polynomials which are in turn connected to the accretivity of certain adjugate polynomial matrices. This turns out to be all the technical vehicles required; the rest of the proof is solely algebraic manipulation. 
%The proof uses a controllability Gramian with rank-one forcing to
%transform the companion matrix into a dissipative matrix. Horner
%polynomials of the transformed matrix are then identified with
%coefficients of an adjugate polynomial. The required mixed positivity
%of these coefficients is established by induction on the matrix
%dimension, using a codimension-one compression and a rank-one
%determinant decomposition. This argument avoids spectral decomposition
%and applies equally to real and nonreal spectra.
% A rank-one
% Lyapunov Gramian changes the companion matrix into a dissipative matrix,
% after which Horner polynomials become coefficients of an adjugate
% polynomial.  We give two complementary proofs of the required mixed
% coefficient positivity: a recursive codimension-one compression and an
% explicit rational zero-exclusion certificate.  These arguments separate
% the algebraic dimension-reduction mechanism from the half-plane analytic
% mechanism.
The argument also gives in a straightforward manner a strict conclusion: if $Q\succ0$, then every
entry of $X$ is strictly positive.  

%This is stronger than $X\succ0$,
%since a positive-definite matrix may have negative off-diagonal
%entries.

Our second purpose is to determine whether the conjecture remains true for the
discrete-time Lyapunov equation
\begin{equation}
  X-A^{\mathsf{T}}XA=Q.
  \label{eq:intro-discrete}
\end{equation}
The answer is simple and blunt: No! 
For a Schur matrix in companion form and a semidefinite $Q\succeq0$, the solution is an observability
Gramian and is therefore positive semidefinite.  Nevertheless, the continuous-time
entrywise conclusion fails already in dimension two, even for
$Q=I$.  A second example shows that the failure persists when all
coefficients of the Schur polynomial are positive.  
% This obstruction
% also shows that the direct disk analogue of the coefficient lemma used
% in the continuous-time proof is false.  
Finally, a Cayley transform
converts \eqref{eq:intro-discrete} into a continuous-time Lyapunov
equation.  After realizing the transformed state matrix in companion
form, the preceding theorem recovers entrywise positivity of a
congruence transform of $X$, although not of $X$ in the original
discrete companion coordinates.

We concede that while the conjecture and its proof are remotely related to systems theory subjects in, e.g., structured
realizations and positive systems (see, for example,
\citet{FarinaRinaldi2000}), they are essentially of academic interest and are driven by intellectual curiosity. We do not purport that the problem will be of significant practical interest.

Section~\ref{sec:problem} states the continuous-time result.
 Section~\ref{sec:adjugate} gives a dimension-reduction proof of the key adjugate lemma.  
% Section~\ref{sec:rational-proof} gives a complementary rational-function proof.  
Section~\ref{sec:continuous-proof} then
proves the main theorem and its strict version.
Section~\ref{sec:discrete} gives the discrete-time results and
counterexamples.  Section~\ref{sec:cayley} develops the transformed
discrete-time statement.

\paragraph{Notation}
For a matrix $X$, we denote its elements by $X_{kl}$. For a real matrix $M$, $M^{\mathsf{T}}$ denotes its transpose.  For a complex
matrix or vector, $M^*$ denotes the conjugate transpose.  The symbols
$M\succeq0$ and $M\succ0$ refer to the Loewner order, indicating that a matrix is nonnegative definite and positive definite, respectively. Define 
\[
  \operatorname{sym}(M):=\frac{M+M^{\mathsf{T}}}{2}.
\]
The open right half-plane and open unit disk are denoted by
$\mathbb{C}_{+}:=\{z\in\mathbb{C}:\operatorname{Re}z>0\}$ and
$\mathbb{D}:=\{z\in\mathbb{C}:|z|<1\}$.  The Frobenius inner product is
$\langle U,V\rangle_F:=\operatorname{tr}(U^{\mathsf{T}}V)$.  
%inequalities are stated
%explicitly and are never denoted by the Loewner symbols.

The proof is partly due to CODEX-5.6 Sol Ultra, and partly due to extensive, informed trainings provided by the authors. To the best of our knowledge, this appears to be the first AI-assisted theoretical result published in the open controls literature, thus witnessing potentially the birth of AI in control systems theory. 

\section{Continuous-time problem and main result}
\label{sec:problem}

Let
\begin{equation*}
  a(s)=s^n+\sigma_1s^{n-1}+\sigma_2s^{n-2}+\cdots+\sigma_n
  \label{eq:hurwitz-polynomial}
\end{equation*}
be a real Hurwitz polynomial, and let
\begin{equation}
  A=
  \begin{bmatrix}
    0&1&&0\\
     &\ddots&\ddots&\\
    0&&0&1\\
    -\sigma_n&-\sigma_{n-1}&\cdots&-\sigma_1
  \end{bmatrix}\in\mathbb{R}^{n\times n}
  \label{eq:companion}
\end{equation}
be its controller companion matrix.  Thus the eigenvalues of $A$ are
the zeros of $a$ and lie in the open left half-plane.
Given $Q=Q^{\mathsf{T}}\succeq0$, consider
\begin{equation}
  XA+A^{\mathsf{T}}X=-Q.
  \label{eq:continuous-lyapunov}
\end{equation}
Standard Lyapunov theory \citep{Kailath1980,HornJohnson2013} gives the
unique solution
\begin{equation*}
  X=\int_0^\infty e^{A^{\mathsf{T}}t}Qe^{At}\,dt\succeq0.
  \label{eq:continuous-integral}
\end{equation*}

% \begin{theorem}[Companion Lyapunov entrywise nonnegativity]
\begin{theorem}
  \label{thm:continuous-main}
  Under the assumptions above, the unique symmetric solution of
  \eqref{eq:continuous-lyapunov} satisfies
  \[
    X_{k\ell}\geq0,
    \qquad 1\leq k,\ell\leq n.
  \]
  % No real-spectrum assumption is required, and the proof does not use
  % a spectral decomposition.
\end{theorem}

The entrywise conclusion is coordinate-dependent: it holds in
the controller companion coordinates of \eqref{eq:companion} and is
not invariant under general similarity transformations.  Nonreal
eigenvalues are allowed through conjugate pairs, but both $A$ and $Q$ must be
real.
% ; no claim is made for companion matrices with complex coefficients.

% \begin{corollary}[Strictly positive-definite forcing]
\begin{corollary}
  \label{cor:strict-continuous}
  Under the assumptions of Theorem~\ref{thm:continuous-main}, if
  $Q\succ0$, then
  \[
    X_{k\ell}>0,
    \qquad 1\leq k,\ell\leq n.
  \]
  Thus $X$ is entrywise strictly positive as well as positive
  definite.
\end{corollary}

% The proof rests on a positivity property of adjugate coefficients.  We
% give two complementary proofs: one by dimension reduction and one by a
% rational zero-exclusion argument.
The proof relies on a mixed positivity property of adjugate
coefficients, which we establish by induction on the matrix dimension.

\section{Adjugate coefficients: a dimension-reduction proof}
\label{sec:adjugate}
We begin with the cofactor compression identity underlying the
dimension-reduction argument.
% We begin with the compression identity shared by both proofs.

% \begin{lemma}[Cofactor compression identity]
\begin{lemma}
  \label{lem:cofactor}
  Let $x\in\mathbb{R}^m$ be a unit vector, and let
  $V\in\mathbb{R}^{m\times(m-1)}$ satisfy
  \[
    V^{\mathsf{T}}V=I,
    \qquad V^{\mathsf{T}}x=0.
  \]
  For arbitrary $M,N\in\mathbb{C}^{m\times m}$,
  \begin{equation}
    x^{\mathsf{T}}\operatorname{adj}(M)\operatorname{adj}(N)^{\mathsf{T}}x
    =\det\!\left(V^{\mathsf{T}}M^{\mathsf{T}}NV\right).
    \label{eq:cofactor-identity}
  \end{equation}
  % Here ${}^{\mathsf{T}}$ is the algebraic transpose, with no complex
  % conjugation.
\end{lemma}

\begin{proof}
Let $U=[\,V\ x\,]$, which is orthogonal.  The identity
$\operatorname{adj}(U^{\mathsf{T}}MU)=U^{\mathsf{T}}\operatorname{adj}(M)U$ shows that, after this change of
basis, it suffices to take $x=e_m$ and
$V=[\,e_1\ \cdots\ e_{m-1}\,]$.  Cauchy--Binet formula \citep{HornJohnson2013} gives
\[
  \det((MV)^{\mathsf{T}}NV)
  =\sum_{j=1}^{m}
    \det((MV)_{\widehat j})\det((NV)_{\widehat j}),
\]
where $\widehat j$ denotes deletion of row $j$. 
% Each summand is the
% product of the minors of $M$ and $N$ obtained by deleting row $j$ and
% column $m$.  After multiplication by the common cofactor sign
% $(-1)^{j+m}$, these are the $j$th entries of the last rows of
% $\operatorname{adj}(M)$ and $\operatorname{adj}(N)$, respectively.  The signs cancel in their
% product, giving the left-hand side of \eqref{eq:cofactor-identity}.
Since $(MV)_{\widehat j}$ is obtained from $M$ by deleting
row $j$ and column $m$, the cofactor formula yields
\[
\bigl(\operatorname{adj}(M)\bigr)_{mj}
=
(-1)^{j+m}\det\!\left((MV)_{\widehat j}\right),
\]
and the same relation holds for $N$. Substituting these
expressions and cancelling the cofactor signs, we obtain
\begin{align*}
\det\!\left(V^{\mathsf{T}}M^{\mathsf{T}}NV\right)
&= \sum_{j=1}^{m}
\bigl(\operatorname{adj}(M)\bigr)_{mj}
\bigl(\operatorname{adj}(N)\bigr)_{mj}\\
&= e_m^{\mathsf{T}}\operatorname{adj}(M)
\operatorname{adj}(N)^{\mathsf{T}}e_m.
\end{align*}

This proves \eqref{eq:cofactor-identity}.
% The argument is polynomial and requires no nonsingularity assumption.
\end{proof}

% \begin{lemma}[Mixed adjugate-coefficient positivity]
\begin{lemma}
  \label{lem:adjugate}
  Let $C\in\mathbb{R}^{m\times m}$ satisfy
  \[
    C+C^{\mathsf{T}}\succeq0.
  \]
  Define $T_0(C),\ldots,T_{m-1}(C)$ by
  \begin{equation*}
    J_C(t):=\operatorname{adj}(I+tC)=\sum_{p=0}^{m-1}t^pT_p(C).
    \label{eq:adjugate-polynomial}
  \end{equation*}
  Then
  \begin{equation}
    \operatorname{sym}\!\left(T_p(C)T_q(C)^{\mathsf{T}}\right)\succeq0,
    \qquad 0\leq p,q\leq m-1.
    \label{eq:mixed-positivity}
  \end{equation}
  % Strict accretivity, invertibility, and diagonalizability are not
  % required.
\end{lemma}

When $p=0$ and $q=m-1$, one has $T_0(C)=I$ and
$T_{m-1}(C)=\operatorname{adj}(C)$.  Thus Lemma~\ref{lem:adjugate} recovers the
accretivity of $\operatorname{adj}(C)$ proved by \citet[Proposition~4.3]{Zhang2026};
the mixed statement for arbitrary $p$ and $q$ is the extension needed
below.

% \begin{proof}[First proof: dimension reduction]
\begin{proof}
We proceed by induction on $m$.  For $m=1$,
$J_C(t)=1$ and the assertion is immediate.  Suppose that the result holds in
dimension $m-1$. Fix a unit vector $x\in\mathbb{R}^m$, and choose $V$ as in
Lemma~\ref{lem:cofactor}.  Set
\begin{equation*}
  B:=V^{\mathsf{T}}CV,
  \qquad
  w:=V^{\mathsf{T}}C^{\mathsf{T}}x.
  \label{eq:induction-B-w}
\end{equation*}
Using $VV^{\mathsf{T}}+xx^{\mathsf{T}}=I$ and
$w^{\mathsf{T}}=x^{\mathsf{T}}CV$, we obtain
\begin{equation}
  CV=VB+xw^{\mathsf{T}},
  \qquad
  B+B^{\mathsf{T}}=V^{\mathsf{T}}(C+C^{\mathsf{T}})V\succeq0.
  \label{eq:induction-compression}
\end{equation}
% The first identity follows from $VV^{\mathsf{T}}+xx^{\mathsf{T}}=I$ and
% $w^{\mathsf{T}}=x^{\mathsf{T}}CV$.
Thus $B$ is a real accretive matrix of size $m-1$.

Let $t,s$ be independent formal variables, and define
\begin{equation*}
  G_x(t,s):=x^{\mathsf{T}}J_C(t)J_C(s)^{\mathsf{T}}x.
  \label{eq:Gx}
\end{equation*}
Applying Lemma~\ref{lem:cofactor} with $M=I+tC$ and $N=I+sC$, and using
\eqref{eq:induction-compression}, yields
\begin{equation*}
  G_x(t,s)
  =\det\!\left((I+tB)^{\mathsf{T}}(I+sB)+tsww^{\mathsf{T}}\right).
  \label{eq:Gx-determinant}
\end{equation*}
Set
\[
  d_B(z):=\det(I+zB),
  \qquad
  J_B(z):=\operatorname{adj}(I+zB).
\]
The rank-one determinant identity
$\det(K+uv^{\mathsf{T}})=\det K+v^{\mathsf{T}}\operatorname{adj}(K)u$ together with
$\operatorname{adj}(PQ)=\operatorname{adj}(Q)\operatorname{adj}(P)$, gives
\begin{equation}
  G_x(t,s)
  =d_B(t)d_B(s)
   +ts\,w^{\mathsf{T}}J_B(s)J_B(t)^{\mathsf{T}}w.
  \label{eq:induction-decomposition}
\end{equation}
% These identities remain valid when the matrices involved are singular,
% because both sides are polynomial in their entries.
Both identities are polynomial in the matrix entries and
remain valid when the matrices involved are singular.

Write
\begin{equation}
  d_B(z)=\sum_{r=0}^{m-1}b_rz^r.
  \label{eq:dB-coefficients}
\end{equation}
Every $b_r$ is nonnegative.  Indeed, if $Bv=\lambda v$ for
$0\neq v\in\mathbb{C}^{m-1}$, then
\[
  2\operatorname{Re}(\lambda)\lVert v\rVert^2
  =v^*(B+B^{\mathsf{T}})v\geq0.
\]
Here the inequality follows by applying
$B+B^{\mathsf{T}}\succeq0$ to the real and imaginary parts
of $v$ separately.
Thus every eigenvalue of $B$ lies in the closed right half-plane.

% Here a real symmetric positive-semidefinite matrix remains
% positive semidefinite on complex vectors, since the quadratic form
% separates into the real and imaginary parts of $v$.
% Hence every eigenvalue of $B$ lies in the closed right half-plane. 
Since $B$ is real, its nonreal eigenvalues occur in conjugate
pairs. In the eigenvalue factorization of $d_B$, each real
eigenvalue $\lambda\geq0$ contributes a factor $1+\lambda z$,
while each nonreal conjugate pair $\alpha\pm i\beta$
contributes
% A real eigenvalue $\lambda$ contributes a factor $1+\lambda z$ with
% $\lambda\geq0$ to $d_B$, while a nonreal conjugate pair
% $\alpha\pm i\beta$ contributes
\[
  1+2\alpha z+(\alpha^2+\beta^2)z^2,
  \qquad \alpha\geq0.
\]
All these factors have nonnegative coefficients, proving
$b_r\geq0$ for every $r$.
% The product of these factors has nonnegative coefficients.  This uses
% only algebraic multiplicities and does not require diagonalizability.
% The rank-one determinant identity, together with
% $\operatorname{adj}(PQ)=\operatorname{adj}(Q)\operatorname{adj}(P)$,
% % both valid also for singular matrices, 
% gives
% \begin{equation}
%   G_x(t,s)
%   =d_B(t)d_B(s)
%    +ts\,w^{\mathsf{T}}J_B(s)J_B(t)^{\mathsf{T}}w.
%   \label{eq:induction-decomposition}
% \end{equation}
% Write
% \begin{equation*}
%   d_B(z)=\sum_{r=0}^{m-1}b_rz^r.
%   \label{eq:dB-coefficients}
% \end{equation*}
% It follows that 
% \begin{equation*}
% d_B(z) = \prod^{m-1}_{i=1}  \lambda_i(I+zB)=\prod^{m-1}_{i=1} (1+z\lambda_i(B)). 
% \end{equation*}
% In view of the fact that $B\succeq 0$, $\lambda_i(B)\geq 0$ for all $i=1,~\cdots,~m-1$, and hence  
% %every eigenvalue of $B$
% %has nonnegative real part. As $B$ is real, its real eigenvalues
% %and nonreal conjugate pairs contribute factors with nonnegative
% %coefficients to $\det(I+zB)$. Hence 
% $b_r\geq0$ for all $r$.

Next, expand
\[
  J_B(z)=\sum_{r=0}^{m-2}z^rT_r(B),
\]
and set $T_r(B)=0$ outside $0\leq r\leq m-2$ and $b_r=0$ outside
$0\leq r\leq m-1$.  Comparing the coefficient of $t^ps^q$ in
\eqref{eq:induction-decomposition}, we obtain, for
$0\leq p,q\leq m-1$,
\begin{equation*}
  x^{\mathsf{T}}T_p(C)T_q(C)^{\mathsf{T}}x
  =b_pb_q
   +w^{\mathsf{T}}T_{q-1}(B)T_{p-1}(B)^{\mathsf{T}}w.
  \label{eq:induction-coefficient}
\end{equation*}
% The first term is nonnegative.  The second is nonnegative by the
% induction hypothesis, since a real quadratic form depends only on the
% symmetric part of its matrix.  Therefore the left-hand side of
The first term on the right-hand side is nonnegative.
The second is nonnegative by the induction hypothesis, since a real quadratic form depends only on the
symmetric part of its matrix.
Consequently,
% \eqref{eq:induction-coefficient} is nonnegative.  Finally,
\[
  x^{\mathsf{T}}T_p(C)T_q(C)^{\mathsf{T}}x
  =x^{\mathsf{T}}\operatorname{sym}\!\left(T_p(C)T_q(C)^{\mathsf{T}}\right)x.
\]
Since $x$ was arbitrary, \eqref{eq:mixed-positivity} follows, completing the 
induction.
\end{proof}

\section{Proof of the continuous-time results}
\label{sec:continuous-proof}

For $0\leq p\leq n-1$, define the Horner polynomials
\begin{equation}
  h_p(s):=s^p+\sigma_1s^{p-1}+\cdots+\sigma_p,
  \qquad h_0(s):=1.
  \label{eq:horner}
\end{equation}
The companion structure yields
\begin{equation}
  h_p(A)e_n=e_{n-p},
  \qquad 0\leq p\leq n-1.
  \label{eq:horner-basis}
\end{equation}
Indeed, this follows recursively from
$h_{p+1}(s)=sh_p(s)+\sigma_{p+1}$ and
$Ae_{n-p}=e_{n-p-1}-\sigma_{p+1}e_n$ for $0\leq p\leq n-2$.
In particular,
$e_n,Ae_n,\ldots,A^{n-1}e_n$ span $\mathbb{R}^n$.

Consider the controllability Gramian
\begin{equation}
  G:=\int_0^\infty
  e^{At}e_ne_n^{\mathsf{T}}e^{A^{\mathsf{T}}t}\,dt.
  \label{eq:gramian}
\end{equation}
It satisfies
\begin{equation}
  AG+GA^{\mathsf{T}}=-e_ne_n^{\mathsf{T}}.
  \label{eq:gramian-equation}
\end{equation}
Equation~\eqref{eq:horner-basis} shows that $(A,e_n)$ is controllable,
because the monicity of the $h_p$ makes the change between the Horner
vectors and $e_n,Ae_n,\ldots,A^{n-1}e_n$ unit triangular.  Hence
$G\succ0$.  Hurwitz stability makes \eqref{eq:gramian} convergent,
and integrating the derivative of its integrand gives
\eqref{eq:gramian-equation}.  Let $R=G^{1/2}=R^{\mathsf{T}}\succ0$ and set
\begin{equation}
  D:=R^{-1}AR,
  \qquad C:=-D.
  \label{eq:balanced-matrices}
\end{equation}
Congruencing \eqref{eq:gramian-equation} by $R^{-1}$ gives
\begin{equation*}
  D+D^{\mathsf{T}}
  =-(R^{-1}e_n)(R^{-1}e_n)^{\mathsf{T}}\preceq0.
  \label{eq:dissipative-D}
\end{equation*}
Consequently, $C+C^{\mathsf{T}}\succeq0$.

Since $D$ is similar to $A$ and $C=-D$,
\begin{equation*}
  \begin{aligned}
    \det(I+tC)
    &=t^n\det(t^{-1}I-D)=t^na(t^{-1})\\
    &=1+\sigma_1t+\cdots+\sigma_nt^n.
  \end{aligned}
  \label{eq:det-reversal}
\end{equation*}
The identity is first obtained for $t\neq0$ and then holds at $t=0$
by polynomial continuation.  Comparing coefficients in
\[
  (I+tC)\operatorname{adj}(I+tC)=\det(I+tC)I
\]
gives, for $1\leq p\leq n-1$,
\begin{equation*}
  \begin{aligned}
    T_0(C)&=I,\\
    T_p(C)&=\sigma_pI-CT_{p-1}(C)\\
          &=\sigma_pI+DT_{p-1}(C).
  \end{aligned}
  \label{eq:T-recurrence}
\end{equation*}
Comparison with \eqref{eq:horner} shows that
\begin{equation*}
  T_p(C)=h_p(D),
  \qquad 0\leq p\leq n-1.
  \label{eq:T-horner}
\end{equation*}
Lemma~\ref{lem:adjugate} therefore implies
\begin{equation}
  \operatorname{sym}\!\left(h_p(D)h_q(D)^{\mathsf{T}}\right)\succeq0,
  \qquad 0\leq p,q\leq n-1.
  \label{eq:horner-positivity}
\end{equation}

\begin{proof}[Proof of Theorem~\ref{thm:continuous-main}]
Fix $1\leq k,\ell\leq n$, and set
\[
  S_{k\ell}:=\operatorname{sym}(e_ke_\ell^{\mathsf{T}}).
\]
Since $A$ is Hurwitz, the Lyapunov operator
$W\mapsto AW+WA^{\mathsf{T}}$ is invertible.  Let $W_{k\ell}$ be the unique
symmetric solution of the adjoint
Lyapunov equation
\begin{equation}
  AW_{k\ell}+W_{k\ell}A^{\mathsf{T}}=-S_{k\ell}.
  \label{eq:adjoint-lyapunov}
\end{equation}
Set $p=n-k$, $q=n-\ell$, and
$Z_{pq}:=h_p(A)Gh_q(A)^{\mathsf{T}}$.  
% Since polynomials in $A$ commute with
Since $h_p(A)$ and $h_q(A)$ commute with
$A$, \eqref{eq:gramian-equation} and \eqref{eq:horner-basis} imply
\begin{equation*}
  \begin{aligned}
    &A\operatorname{sym}(Z_{pq})
      +\operatorname{sym}(Z_{pq})A^{\mathsf{T}}\\
    &\quad=\operatorname{sym}\!\left(
      h_p(A)(AG+GA^{\mathsf{T}})h_q(A)^{\mathsf{T}}\right)\\
    &\quad=-\operatorname{sym}(e_ke_\ell^{\mathsf{T}})=-S_{k\ell}.
  \end{aligned}
\end{equation*}
Uniqueness in \eqref{eq:adjoint-lyapunov} and
\eqref{eq:balanced-matrices} thus yield
\begin{align}
  W_{k\ell}
  &=\operatorname{sym}\!\left(h_p(A)Gh_q(A)^{\mathsf{T}}\right)\notag\\
  &=R\operatorname{sym}\!\left(h_p(D)h_q(D)^{\mathsf{T}}\right)R
    \succeq0
  \label{eq:W-positive}
\end{align}
by \eqref{eq:horner-positivity}.

Taking the Frobenius inner product of
\eqref{eq:continuous-lyapunov} with $W_{k\ell}$ and using
\eqref{eq:adjoint-lyapunov}, we obtain
\begin{align*}
  -\operatorname{tr}(QW_{k\ell})
  &=\langle XA+A^{\mathsf{T}}X,W_{k\ell}\rangle_F\\
  &=\langle X,AW_{k\ell}+W_{k\ell}A^{\mathsf{T}}\rangle_F\\
  &=-\langle X,S_{k\ell}\rangle_F=-X_{k\ell}.
\end{align*}
Therefore
\begin{equation}
  X_{k\ell}=\operatorname{tr}(QW_{k\ell})=\operatorname{tr}(Q^{1/2}W_{k\ell}Q^{1/2})\geq0.
  \label{eq:entry-trace}
\end{equation}
%because $Q,W_{k\ell}\succeq0$.  This holds for every $k,\ell$.
\end{proof}

\begin{proof}[Proof of Corollary~\ref{cor:strict-continuous}]
Fix $1\leq k,\ell\leq n$.
Since $S_{k\ell}\neq0$, equation
\eqref{eq:adjoint-lyapunov} implies that $W_{k\ell}\neq0$.
% The matrix $W_{k\ell}$ in \eqref{eq:adjoint-lyapunov} is nonzero,
% because $S_{k\ell}\neq0$.  
Together with \eqref{eq:W-positive}, this gives
$\operatorname{tr}(W_{k\ell})>0$.
% It is also positive semidefinite by
% \eqref{eq:W-positive}, and therefore $\operatorname{tr}(W_{k\ell})>0$.  
If
$Q\succ0$, then
$Q\succeq\lambda_{\min}(Q)I$ with $\lambda_{\min}(Q)>0$. Hence \eqref{eq:entry-trace} yields 
% it follows from the von Neumann inequality \cite{HornJohnson2013} that 
\[
  X_{k\ell} = \operatorname{tr}(QW_{k\ell})
  =\operatorname{tr}\!\left(W_{k\ell}^{1/2}QW_{k\ell}^{1/2}\right)
  \geq\lambda_{\min}(Q)\operatorname{tr}(W_{k\ell})>0.
\]
Since $k$ and $\ell$ were arbitrary, every entry of $X$
is strictly positive.
% The proof then follows, likewise, from \eqref{eq:entry-trace}.
\end{proof}

\section{The discrete-time Lyapunov equation}
\label{sec:discrete}

Let $A\in\mathbb{R}^{n\times n}$ be Schur stable, i.e., $\rho(A)<1$, and
consider the Stein equation
\begin{equation}
  X-A^{\mathsf{T}}XA=Q,
  \qquad Q=Q^{\mathsf{T}}.
  \label{eq:stein}
\end{equation}
This orientation is the direct discrete counterpart of
\eqref{eq:continuous-lyapunov}.  The convention
$X-AXA^{\mathsf{T}}=Q$ is obtained by replacing $A$ with $A^{\mathsf{T}}$.
The standard Gramian properties recalled below can be found,
% The standard Gramian facts below may be found, 
for example, in
\citet{ZhouDoyleGlover1996}. 

% \begin{proposition}[What always survives]
\begin{proposition}[\citet{ZhouDoyleGlover1996}]
  \label{prop:stein-standard}
  Equation~\eqref{eq:stein} has the unique symmetric solution
  \begin{equation}
    X=\sum_{j=0}^{\infty}(A^{\mathsf{T}})^jQA^j.
    \label{eq:stein-series}
  \end{equation}
  Consequently,
  \[
    Q\succeq0\Longrightarrow X\succeq0,
    \qquad
    Q\succ0\Longrightarrow X\succeq Q\succ0.
  \]
  For $Q\succeq0$, one has $X\succ0$ if and only if the pair
  $(Q^{1/2},A)$ is observable.
\end{proposition}

% \begin{proof}
% Since $A$ is Schur stable, the series in
% \eqref{eq:stein-series} converges in norm to a symmetric matrix.
% % Schur stability makes the series \eqref{eq:stein-series} convergent.
% % Shifting its index verifies \eqref{eq:stein}.  If two solutions
% % existed, their difference $Z$ would satisfy
% % $Z=(A^{\mathsf{T}})^jZA^j$ for every $j$, whose right-hand side tends to zero;
% % hence $Z=0$.  
% An index shift verifies \eqref{eq:stein}.
% If $Z$ is the difference of two solutions, then
% \[
% Z=(A^{\mathsf{T}})^jZA^j,
% \qquad\text{for every }j\geq0.
% \]
% The right-hand side tends to zero as $j\to\infty$, so $Z=0$.

% For $Q\succeq0$ and $v\in\mathbb{R}^n$,
% \begin{equation}
%   v^{\mathsf{T}}Xv
%   =\sum_{j=0}^{\infty}
%     \lVert Q^{1/2}A^jv\rVert^2\geq0.
%   \label{eq:stein-quadratic}
% \end{equation}
% The right-hand side vanishes if and only if
% $Q^{1/2}A^jv=0$ for every $j\geq0$.
% By the Cayley--Hamilton theorem, this is equivalent to
% $Q^{1/2}A^jv=0$ for $0\leq j\leq n-1$.
% Thus $X\succ0$ if and only if the only vector satisfying these
% conditions is $v=0$, which is precisely the observability
% criterion.

% % If $Q\succ0$, the $j=0$ term is strictly positive for $v\neq0$.
% % For semidefinite $Q$, equality in \eqref{eq:stein-quadratic} holds
% % exactly when $Q^{1/2}A^jv=0$ for every $j$.  By Cayley--Hamilton it is
% % enough to check $0\leq j\leq n-1$, which is precisely the
% % observability criterion.
% \end{proof}
The following examples show that no general entrywise-nonnegativity
conclusion can be added to Proposition~\ref{prop:stein-standard}, even
for an identity forcing matrix.

% \begin{proposition}[Minimal discrete-time counterexample]
\begin{example}
  \label{prop:quadratic-counterexample}
  Let
  \begin{equation}
    a_2(z)=z^2+\frac12z-\frac14,
    \qquad
    A_2=
    \begin{bmatrix}
      0&1\\[1mm]
      \frac14&-\frac12
    \end{bmatrix}.
    \label{eq:quadratic-example}
  \end{equation}
  Then $a_2$ is strictly Schur stable, but the solution of
  $X-A_2^{\mathsf{T}}XA_2=I_2$ is
  \begin{equation}
    X_2=\frac1{25}
    \begin{bmatrix}
      31&-16\\
      -16&96
    \end{bmatrix}\succ0.
    \label{eq:X2}
  \end{equation}
  In particular, $(X_2)_{12}=-16/25<0$.  Dimension two is the smallest
  dimension in which this can occur.
\end{example}

\begin{proof}
The zeros of $a_2$ are $(-1\pm\sqrt5)/4$, both of which
have modulus less than one. Direct multiplication gives
\[
  A_2^{\mathsf{T}}X_2A_2
  =\frac1{25}
    \begin{bmatrix}
      6&-16\\
      -16&71
    \end{bmatrix},
\]
and therefore $X_2-A_2^{\mathsf{T}}X_2A_2=I_2$.  Uniqueness and positive definiteness
follow from Proposition~\ref{prop:stein-standard}.  

% In dimension one,
% the equation is $X(1-A^2)=Q$, so $X=Q/(1-A^2)$ has the same scalar sign
% as $Q$ whenever $|A|<1$.
In dimension one, the solution is
$X=Q/(1-A^2)$.
Since $1-A^2>0$ whenever $|A|<1$, a nonnegative forcing
$Q$ necessarily gives $X\geq0$.
This proves the minimality of dimension two.
\end{proof}

% \begin{remark}[The alternate Stein convention]
\begin{remark}
  For the convention $Y-A_2YA_2^{\mathsf{T}}=I_2$, the same controller
  companion matrix in \eqref{eq:quadratic-example} gives
  \[
    Y=\frac1{25}
    \begin{bmatrix}
      76&-34\\
      -34&51
    \end{bmatrix}\succ0.
  \]
  Its off-diagonal entries are negative.
Thus entrywise nonnegativity can fail under either standard
Stein convention.
  % Thus failure of entrywise nonnegativity occurs under either of the
  % two standard Stein orientations.
\end{remark}

% \begin{remark}[A family of quadratic counterexamples]
\begin{remark}
  \label{rem:quadratic-family}
  For
  \[
    a(z)=z^2+\alpha z+\beta,
    \qquad
    A=\begin{bmatrix}0&1\\-\beta&-\alpha\end{bmatrix},
  \]
  solving $X-A^{\mathsf{T}}XA=I_2$ entry by entry gives
  \begin{equation}
    X_{12}
    =\frac{2\alpha\beta}
    {(1-\beta)((1+\beta)^2-\alpha^2)}.
    \label{eq:quadratic-formula}
  \end{equation}
  The strict Jury conditions imply
  $1-\beta>0$ and $1+\beta>|\alpha|$, so the denominator in
  \eqref{eq:quadratic-formula} is positive.  Hence the sign of
  $X_{12}$ is the sign of $\alpha\beta$, and every strictly Schur pair
  with $\alpha\beta<0$ gives a counterexample.
\end{remark}

% The preceding quadratic has coefficients of mixed sign.  This is not
% the source of the obstruction.
Entrywise nonnegativity can also fail when every coefficient
of the defining polynomial is strictly positive.

% \begin{proposition}[Counterexample with positive coefficients]
\begin{example}
  \label{prop:positive-coefficient-counterexample}
  The polynomial and its companion matrix
  \begin{equation*}
    \begin{split}
    a_3(z)&=z^3+\frac1{10}z^2+\frac15z+\frac35,\\
    A_3&=
    \begin{bmatrix}
      0&1&0\\
      0&0&1\\
      -\frac35&-\frac15&-\frac1{10}
    \end{bmatrix}
    \end{split}
    \label{eq:cubic-example}
  \end{equation*}
  are strictly Schur stable.  For $Q=I_3$, the Stein solution is
  \begin{equation}
    X_3=\frac1{2375}
    \begin{bmatrix}
      6587&1512&-180\\
      1512&9502&1470\\
      -180&1470&11700
    \end{bmatrix}\succ0,
    \label{eq:X3}
  \end{equation}
  and $(X_3)_{13}=-36/475<0$.
\end{example}

\begin{proof}
On $|z|=1$,
\[
  \left|\frac1{10}z^2+\frac15z+\frac35\right|
  \leq\frac9{10}<1=|z^3|.
\]
% Rouch\'e's theorem shows that $a_3$ and $z^3$ have the same number of
% zeros in $\mathbb{D}$, so all three zeros of $a_3$ lie strictly in $\mathbb{D}$.
By Rouch\'e's theorem, $a_3$ and $z^3$ have the same
number of zeros in the open unit disk $\mathbb{D}$.
Thus all three zeros of $a_3$ lie in $\mathbb{D}$,
and $A_3$ is Schur stable.
Direct multiplication gives
\[
  A_3^{\mathsf{T}}X_3A_3
  =\frac1{2375}
  \begin{bmatrix}
    4212&1512&-180\\
    1512&7127&1470\\
    -180&1470&9325
  \end{bmatrix}
  =X_3-I_3.
\]
Hence $X_3$ solves the Stein equation.
Proposition~\ref{prop:stein-standard} gives uniqueness
and positive definiteness, while \eqref{eq:X3} exhibits
the claimed negative entry.
% Thus \eqref{eq:X3} is the unique Stein solution and has the claimed
% negative entry.  It is positive definite by
% Proposition~\ref{prop:stein-standard}.
\end{proof}

% \begin{remark}[Hardy-space interpretation]
%   \label{rem:hardy}
%   Factor $Q=L^{\mathsf{T}}L$ and define
%   \[
%     f_k(w):=L(I-wA)^{-1}e_k
%              =\sum_{j=0}^{\infty}w^jLA^je_k.
%   \]
%   Then \eqref{eq:stein-series} says
%   \[
%     X_{k\ell}=\langle f_k,f_\ell\rangle_{H^2}.
%   \]
%   Thus $X$ is a Gram matrix, which explains its Loewner positivity;
%   cross inner products, however, have no prescribed sign.  At the
%   polynomial level, the direct disk analogue of the half-plane
%   coefficient lemma is false: $P(t,s)=1-\tfrac12ts$ is zero-free on
%   $\mathbb{D}^2$ and positive on $(0,1)^2$, but has a negative coefficient.
%   This example shows that disk stability and positivity on $(0,1)^2$
%   do not by themselves imply coefficientwise nonnegativity; it does
%   not exclude more specialized certificates using companion structure.
% \end{remark}
% \begin{remark}[Hardy-space interpretation]
\begin{remark}
\label{rem:hardy}
For $Q\succeq0$, choose $L\in\mathbb{R}^{r\times n}$ such that
$Q=L^{\mathsf{T}}L$, and define, for $w\in\mathbb{D}$,
\[
f_k(w):=L(I-wA)^{-1}e_k
=
\sum_{j=0}^{\infty}w^jLA^je_k.
\]
Schur stability ensures that these functions belong to the
vector-valued Hardy space $H^2(\mathbb{D};\mathbb{C}^r)$.
Using its coefficient inner product,
\eqref{eq:stein-series} gives
\[
X_{k\ell}=\langle f_k,f_\ell\rangle_{H^2}.
\]
Thus $X$ is the Gram matrix of $f_1,\ldots,f_n$, which
explains its positive semidefiniteness.
The inner products between distinct functions may be negative.
% For $Q\succeq0$, factor $Q=L^{\mathsf{T}}L$ and define
% \[
%   f_k(w):=L(I-wA)^{-1}e_k
%   =\sum_{j=0}^{\infty}w^jLA^je_k.
% \]
% Then \eqref{eq:stein-series} gives
% \[
%   X_{k\ell}=\langle f_k,f_\ell\rangle_{H^2}.
% \]
% Thus $X$ is a Gram matrix, which explains its positive semidefiniteness.
% The inner products between distinct functions, however, need not be
% nonnegative.
\end{remark}
\section{Positivity after a Cayley--companion transformation}
\label{sec:cayley}

Propositions~\ref{prop:quadratic-counterexample}
and~\ref{prop:positive-coefficient-counterexample} show that
a general entrywise nonnegativity result fails in the original
discrete companion coordinates. The continuous-time theorem
nevertheless yields such a result after a Cayley transform
and a change of basis.

% Although Propositions~\ref{prop:quadratic-counterexample}
% and~\ref{prop:positive-coefficient-counterexample} rule out the direct
% discrete analogue, a coordinate-transformed version follows from the
% continuous-time theorem.

% \begin{theorem}[Cayley--companion positivity]
\begin{theorem}
  \label{thm:cayley}
  Let $A\in\mathbb{R}^{n\times n}$ be a Schur stable controller companion
  matrix in the form \eqref{eq:companion}, let
  $Q=Q^{\mathsf{T}}\succeq0$, and let $X$ be the solution
of \eqref{eq:stein}.  Define
  \begin{equation}
    H:=(A-I)(A+I)^{-1}.
    \label{eq:cayley-H}
  \end{equation}
  There exists a real nonsingular matrix $S$, depending only on $A$,
  such that
  \[
    B:=S^{-1}HS
  \]
  is a Hurwitz controller companion matrix and
  \begin{equation*}
    Y:=S^{\mathsf{T}}XS
    \label{eq:Y-congruence}
  \end{equation*}
  is entrywise nonnegative.  If $Q\succ0$, then every entry of $Y$ is
  strictly positive.
\end{theorem}

\begin{proof}
Since $A$ is strictly Schur, $-1$ is not an eigenvalue of $A$, so
\eqref{eq:cayley-H} is well defined.  By the spectral mapping theorem, every eigenvalue of $H$
has the form
% Spectral mapping sends an
% eigenvalue $\lambda\in\mathbb{D}$ of $A$ to
\[
\mu=\frac{\lambda-1}{\lambda+1},
\qquad \lambda\in\sigma(A).
\]
Moreover,
\[
\operatorname{Re}\mu
=
\frac{|\lambda|^2-1}{|\lambda+1|^2}
<0.
\]
% \[
%   \mu=\frac{\lambda-1}{\lambda+1},
%   \qquad \operatorname{Re}\mu<0,
% \]
Thus $H$ is Hurwitz.  The inverse Cayley relation is
\begin{equation}
  A=(I+H)(I-H)^{-1}.
  \label{eq:inverse-cayley}
\end{equation}
By the Cayley--Hamilton theorem, $(A+I)^{-1}$ and
$(I-H)^{-1}$ are polynomials in $A$ and $H$, respectively.
Hence the real unital matrix algebras generated by $A$
and $H$ coincide.
Since $A$ has the controller companion form
\eqref{eq:companion}, $e_n$ is a cyclic vector for $A$.
Consequently,
% The inverses in \eqref{eq:cayley-H} and
% \eqref{eq:inverse-cayley} are polynomials in the corresponding
% matrices by Cayley--Hamilton.  Hence the unital matrix algebras
% generated by $A$ and $H$ coincide.  A companion matrix is cyclic, so
% $e_n$ is a cyclic vector for $A$ and
\begin{align*}
  \operatorname{span}\{p(H)e_n:p\in\mathbb{R}[s]\}
  &=\operatorname{span}\{p(A)e_n:p\in\mathbb{R}[s]\}\\
  &=\mathbb{R}^n.
\end{align*}
% Thus $H$ is cyclic and is similar over $\mathbb{R}$ to the controller
% companion matrix $B$ of its characteristic polynomial.  This gives a
% real nonsingular $S$ with $B=S^{-1}HS$.
Thus $e_n$ is also cyclic for $H$, and $H$ is similar
over $\mathbb{R}$ to the controller companion matrix $B$
of its characteristic polynomial.
We may therefore choose a real nonsingular matrix $S$,
depending only on $A$, such that $B=S^{-1}HS$.

Using \eqref{eq:inverse-cayley}, multiply \eqref{eq:stein} on the left
by $(I-H)^{\mathsf{T}}$ and on the right by $I-H$.  Expansion gives
\begin{equation}
  H^{\mathsf{T}}X+XH
  =-\frac12(I-H)^{\mathsf{T}}Q(I-H).
  \label{eq:stein-to-lyapunov}
\end{equation}
Congruencing \eqref{eq:stein-to-lyapunov} by $S$ gives
\begin{equation}
  B^{\mathsf{T}}Y+YB=-\widehat Q,
  \qquad
  \widehat Q:=\frac12
  S^{\mathsf{T}}(I-H)^{\mathsf{T}}Q(I-H)S\succeq0.
  \label{eq:transformed-lyapunov}
\end{equation}
Theorem~\ref{thm:continuous-main} applied to $B$ proves that $Y$ is
entrywise nonnegative.  If $Q\succ0$, then $\widehat Q\succ0$ because
$(I-H)S$ is nonsingular, and
Corollary~\ref{cor:strict-continuous} makes every entry of $Y$ strictly
positive.
\end{proof}

% \begin{remark}[A constructive choice of coordinates]
\begin{remark}
  \label{rem:constructive-S}
  The matrix $S$ can be constructed explicitly from $A$. 
  Let $a_d(z):=\det(zI-A)$ and write
  % be the discrete-time Schur polynomial and
  % let
  \[
    b(s):=\det(sI-H)
    =\frac{(1-s)^n}{\det(A+I)}
      a_d\!\left(\frac{1+s}{1-s}\right).
  \]
  The right-hand side is interpreted as a polynomial after
cancellation, so the identity also holds at $s=1$.
Write
  $b(s)=s^n+\beta_1s^{n-1}+\cdots+\beta_n$,
  and define its Horner polynomials by
\[
k_0(s):=1,~~
k_p(s):=s^p+\beta_1s^{p-1}+\cdots+\beta_p,
\quad 1\leq p\leq n-1.
\]
Then define $S$ through
\[
Se_{n-p}=k_p(H)e_n,
\qquad 0\leq p\leq n-1.
\]
Since $e_n$ is cyclic for $H$ and each $k_p$ is monic
of degree $p$, these columns form a basis of $\mathbb{R}^n$,
so $S$ is nonsingular.
The Horner recurrence, together with the Cayley--Hamilton
theorem, gives $HS=SB$, where $B$ is the controller companion
matrix of $b$.
Thus $S$ depends only on $A$, or equivalently on $a_d$
in the fixed controller realization, and is independent of $Q$.
  % If $k_p(s)=s^p+\beta_1s^{p-1}+\cdots+\beta_p$, then one may choose
  % $S$ through the Horner basis
  % \[
  %   Se_{n-p}=k_p(H)e_n,
  %   \qquad 0\leq p\leq n-1.
  % \]
  % Indeed, $e_n$ is cyclic for $H$ because it is cyclic for $A$ and the
  % algebras generated by $A$ and $H$ coincide.  This makes clear that
  % $S$ depends only on $A$ (equivalently, on its polynomial in the fixed
  % controller realization), not on $Q$.
\end{remark}

\begin{remark}
The conclusion of Theorem~\ref{thm:cayley} concerns
$Y=S^{\mathsf{T}}XS$, which represents the same quadratic form
as $X$ in a new state basis.
Entrywise nonnegativity is not preserved under general
congruences, so this conclusion is consistent with the negative
entries in \eqref{eq:X2} and \eqref{eq:X3}.
In particular, the theorem does not imply entrywise
nonnegativity of $X$ in the original discrete companion
coordinates.
% The conclusion of Theorem~\ref{thm:cayley} concerns the congruence
% $S^{\mathsf{T}}XS$, which is the matrix of the same quadratic form in a new
% state basis.  Entrywise order is not invariant under congruence.  The
% theorem is therefore compatible with the negative entries in
% \eqref{eq:X2} and \eqref{eq:X3}; it does not restore entrywise
% positivity of $X$ in the original discrete companion coordinates.
\end{remark}

\section{Conclusion}
\label{sec:conclusion}

For real Hurwitz controller companion matrices, we have shown that
positive-semidefinite forcing produces an entrywise nonnegative
Lyapunov solution, while positive-definite forcing makes every entry
strictly positive. The proof relies on mixed positivity of adjugate
coefficients, established by dimension reduction and transferred to
the Lyapunov equation through a controllability Gramian and the Horner
basis. These entrywise conclusions are specific to the controller
companion coordinates.

The direct discrete-time analogue fails even for identity forcing
and positive polynomial coefficients. However, a Cayley transformation
recovers entrywise nonnegativity in a transformed companion basis that
depends only on the state matrix. An open direction is to characterize
the Schur companion matrices and positive-semidefinite forcing matrices
for which the Stein solution is entrywise nonnegative in the original
coordinates.
% The continuous-time companion structure imposes an entrywise order on
% the entire Lyapunov solution operator: positive-semidefinite forcing
% produces an entrywise nonnegative solution, and positive-definite
% forcing produces an entrywise strictly positive solution.  The common
% key in our two proofs is the mixed positivity of the adjugate
% coefficients.  The first proof identifies its algebraic source in a
% codimension-one compression and a rank-one determinant decomposition;
% the alternative proof identifies its analytic source in the
%   half-plane geometry encoded by the rational certificate
% \eqref{eq:rational-certificate}.  Discrete-time stability retains the
% usual Gramian and Loewner-order conclusions but does not retain the
% entrywise order, even for identity forcing and positive polynomial
% coefficients.  A Cayley transform identifies a coordinate-transformed
% remnant of the continuous theorem: entrywise positivity holds after
% passing to a companion basis for the transformed Hurwitz matrix.  Further work may
% seek intrinsic characterizations of discrete companion polynomials
% and forcing matrices for which entrywise nonnegativity happens to hold
% in the original coordinates.

\bibliographystyle{elsarticle-harv}
\bibliography{companion_lyapunov_references_revised}

\end{document}